\documentclass[a4paper,11pt]{amsart}
\usepackage[utf8x,utf8]{inputenc} 
\usepackage[T1]{fontenc}
\usepackage{lmodern}
\usepackage{url}
\usepackage{amsfonts}
\usepackage{graphicx}

\usepackage{amssymb}
\usepackage{amsmath}
\usepackage{amsthm}
\usepackage[foot]{amsaddr}

\usepackage{euscript}

\usepackage{microtype}

\newtheorem{theorem}{Theorem}
\newtheorem{lemma}[theorem]{Lemma}

\theoremstyle{definition}

\newtheorem{problem}[theorem]{Problem}

\newtheorem{algorithm}[theorem]{Algorithm}

\theoremstyle{remark}

\newcommand{\N}{\mathbb{N}} 
\newcommand{\R}{\mathbb{R}} 

\newcommand{\F}{\EuScript{F}} 

\newcommand{\ceil}[1]{\lceil {#1} \rceil}
\newcommand{\floor}[1]{\lfloor {#1} \rfloor}
\newcommand{\card}[1]{\left\lvert {#1} \right\rvert}
\newcommand{\ind}[1]{\left\langle {#1} \right\rangle}

\newcommand{\Set}[2]{{\left\{ {#1} \colon {#2} \right\}}}

\newcommand{\downc}{\mathord{\downarrow}}

\newcommand{\sgsum}{\bigoplus}
\newcommand{\lub}{\vee}
\newcommand{\setchoose}[2]{\left(\!\genfrac{}{}{0pt}{}{#1}{#2}\!\right)}
\newcommand{\bv}[1]{{\{ 0, 1 \}^{#1}}}
\newcommand{\bvc}[2]{{\setchoose{ \{ 0, 1 \}^{#1} }{\downc {#2}}}}

\usepackage{enumitem}
\usepackage{hyperref}
\hypersetup{bookmarks, pdftex}

\begin{document}

\title{Fast Monotone Summation over Disjoint Sets$^*$}
\thanks{$^*$This paper is an extended version of a conference abstract by the present authors \cite{kaski2012disjoint}. The research was supported in part by the Academy of Finland, 
Grants 252083 (P.K.), 256287 (P.K.), and 125637 (M.K.),
and by the Helsinki Doctoral Programme in Computer Science - Advanced Computing and Intelligent Systems (J.K.).}

\author{Petteri Kaski$^1$}
\address{$^1$Helsinki Institute for Information Technology HIIT \and
Department of Information and Computer Science, Aalto University\\
Finland}
\email{petteri.kaski@aalto.fi}

\author{Mikko Koivisto$^2$}
\email{mikko.koivisto@cs.helsinki.fi}

\author{Janne H. Korhonen$^2$}
\address{$^2$Helsinki Institute for Information Technology HIIT \and
Department of Computer Science, University of Helsinki\\
Finland}
\email{janne.h.korhonen@cs.helsinki.fi}



\begin{abstract}
We study the problem of computing an ensemble of multiple 
sums where the summands in each sum are indexed by subsets 
of size $p$ of an $n$-element ground set. 
More precisely, the task is to compute, 
for each subset of size $q$ of the ground set, the 
sum over the values of all subsets 
of size $p$ that are {\em disjoint} from the subset of size $q$. 
We present an arithmetic circuit that, without subtraction, solves the problem 
using $O((n^p+n^q)\log n)$ arithmetic gates, all monotone;  
for constant $p$, $q$ this is within the factor $\log n$ of the
optimal. The circuit design is based on viewing 
the summation as a ``set nucleation'' task and using a 
tree-projection approach to implement the nucleation. 
Applications include improved algorithms for counting heaviest $k$-paths in a weighted
graph, computing permanents of rectangular matrices, and dynamic feature selection 
in machine learning.
\end{abstract}

\maketitle

\section{Introduction}

\subsection{Weak algebrisation.} 
Many hard combinatorial problems benefit from {\em algebrisation},
where the problem to be solved is cast in algebraic terms as the task 
of evaluating a particular expression or function over a suitably rich 
algebraic structure, such as a multivariate polynomial ring over 
a finite field. Recent
advances in this direction include improved algorithms for 
the $k$-path \cite{williams2009finding}, 
Hamiltonian path \cite{bjorklund2010determinant}, 
$k$-coloring \cite{bjorklund2009set}, 
Tutte polynomial \cite{bjorklund2008computing}, 
knapsack \cite{lokshtanov2010saving},
and connectivity \cite{cygan2011solving} problems.
A key ingredient in all of these advances is the exploitation
of an algebraic catalyst, such as the existence of additive inverses 
for inclusion--exclusion, or the existence of roots of unity for
evaluation/interpolation, to obtain fast evaluation algorithms.

Such advances withstanding, it is a basic question whether 
the catalyst is {\em necessary} to obtain speedup. For example,
fast algorithms for matrix multiplication 
\cite{cohn2005grouptheoretic,coppersmith1990matrix}
(and combinatorially related tasks such as finding a triangle 
in a graph \cite{alon1997finding,itai1978finding}) rely
on the assumption that the scalars have a ring structure,
which prompts the question whether a weaker structure, such as
a semiring without additive inverses, would still enable 
fast multiplication. The answer to this particular question is 
known to be negative \cite{kerr1970effect}, but for many of
the recent advances such an analysis has not been carried out.
In particular, many of the recent algebrisations have 
significant combinatorial structure, which gives hope for 
{\em positive} results even if algebraic catalysts are lacking. 
The objective of this paper is to present one such positive result
by deploying {\em combinatorial} tools.

\subsection{A lemma of Valiant.} 
Our present study stems from a technical lemma of 
Valiant \cite{valiant1986negation} encountered in the study of circuit 
complexity over a monotone versus a universal basis. 
More specifically, starting from $n$ variables 
$f_1,f_2,\ldots,f_n$, the objective is to use as few arithmetic 
operations as possible to compute the $n$ sums of variables where 
the $j$th sum $e_j$ includes all the other variables except the variable $f_j$, 
where $j=1,2,\ldots,n$. 

If additive inverses are available, a solution using $O(n)$ arithmetic 
operations is immediate: first take the sum of all the $n$ variables, 
and then for $j=1,2,\ldots,n$ compute $e_j$ by subtracting 
the variable $f_j$. 

Valiant~\cite{valiant1986negation} showed that $O(n)$ operations suffice 
also when additive inverses are {\em not} available; we display 
Valiant's elegant combinatorial solution for $n=8$ below as an arithmetic 
circuit. [[Please see Appendix~\ref{appendix:valiant} for the general case.]]

\begin{center}
\includegraphics{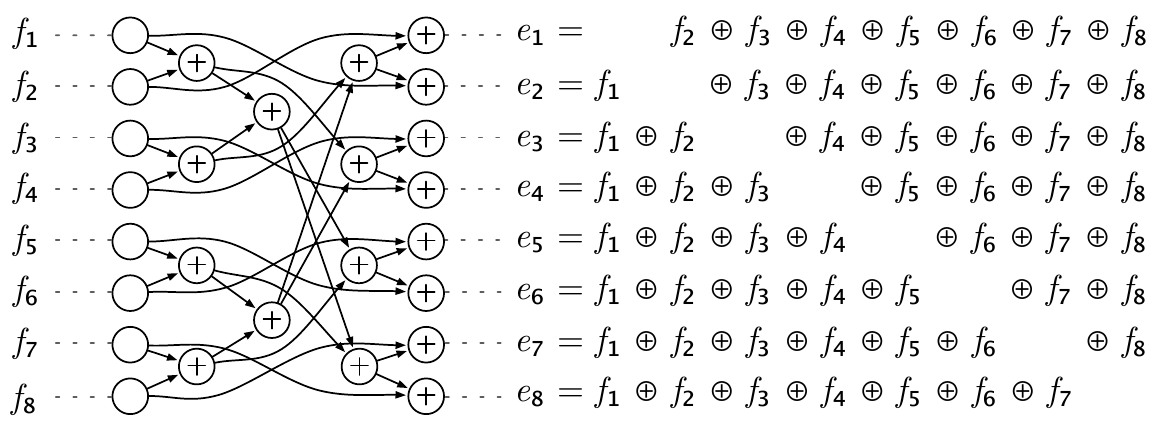}
\end{center}

\subsection{Generalising to higher dimensions.} 
This paper generalises Valiant's lemma to higher
dimensions using purely combinatorial tools. Accordingly, we 
assume that only very limited algebraic structure is available
in the form of a commutative semigroup $(S,\oplus)$. 
That is, $\oplus$ satisfies the associative law 
$x\oplus(y\oplus z)=(x\oplus y)\oplus z$ and the commutative 
law $x\oplus y=y\oplus x$ for all $x,y,z\in S$, but nothing 
else is assumed.

By ``higher dimensions'' we refer to the input not consisting of
$n$ values (``variables'' in the example above) in $S$, but rather 
$\binom{n}{p}$ values $f(X)\in S$ indexed by the $p$-subsets $X$ of
$[n]=\{1,2,\ldots,n\}$. Accordingly, we also allow the output
to have higher dimension. That is, given as input a function $f$ from 
the $p$-subsets $[n]$ to the set $S$, the task is to output 
the function $e$ defined for each $q$-subset $Y$ of $[n]$ by
\begin{equation}\label{eq:e}
e(Y)=\bigoplus_{X:X\cap Y=\emptyset}f(X)\,,
\end{equation}
where the sum is over all $p$-subsets $X$ of $[n]$ satisfying the 
intersection constraint. Let us call this problem 
{\em $(p,q)$-disjoint summation}.
 
In analogy with Valiant's solution for the case $p=q=1$
depicted above, an algorithm that solves the $(p,q)$-disjoint 
summation problem can now be viewed as a circuit consisting of two types 
of gates: {\em input gates} indexed
by $p$-subsets $X$ and {\em arithmetic gates} that perform the 
operation $\oplus$, with certain arithmetic gates designated 
as output gates indexed by $q$-subsets $Y$.
We would like a circuit that has as few gates as possible. 
In particular, does there exist a circuit whose size for 
constant $p$, $q$ is within a logarithmic factor of 
the lower bound $\Theta(n^p + n^q)$? 

\subsection{Main result.} In this paper we answer 
the question in the affirmative. Specifically, we show that a circuit 
of size $O\bigl((n^p+n^q)\log n\bigr)$ exists to compute $e$ from $f$
over an arbitrary commutative semigroup $(S,\oplus)$, and moreover, 
there is an algorithm that constructs the circuit in time 
$O\bigl((p^2+q^2)(n^p+n^q)\log^3 n\bigr)$. 
These bounds hold uniformly for all $p$, $q$. That is,
the coefficient hidden by $O$-notation does not depend on $p$ and $q$.

From a technical perspective our main contribution is 
combinatorial and can be expressed as a solution to 
a specific {\em set nucleation} task. In such 
a task we start with a collection of ``atomic compounds'' 
(a collection of singleton sets), and the goal is to assemble 
a specified collection of ``target compounds'' (a collection of sets
that are unions of the singletons). The assembly is to be executed
by a straight-line program, where each operation in the program 
selects two {\em disjoint} sets in the collection and inserts 
their union into the collection. 
(Once a set is in the collection, it may be selected arbitrarily 
many times.) The assembly should be done in as few operations 
as possible. 

Our main contribution can be viewed as a straight-line program 
of length $O\bigl((n^p+n^q)\log n\bigr)$ that assembles the
collection $\{\{X:X\cap Y=\emptyset\}:Y\}$ starting
from the collection $\{\{X\}:X\}$, where $X$ ranges
over the $p$-subsets of $[n]$ and $Y$ ranges over
the $q$-subsets of $[n]$. 
Valiant's lemma~\cite{valiant1986negation} in these terms 
provides an optimal solution of length $\Theta(n)$ for 
the specific case $p=q=1$. 

\subsection{Applications.} 
Many classical optimisation problems and counting problems can be algebrised
over a commutative semigroup. A selection of applications will be 
reviewed in Sect.~\ref{sect:applications}.

\subsection{Related work.} ``Nucleation'' is implicit in the design of many fast algebraic 
algorithms, perhaps two of the most central are the fast Fourier
transform of Cooley and Tukey \cite{cooley1965algorithm} 
(as is witnessed by the butterfly circuit representation) and
Yates's 1937 algorithm \cite{yates1937design} 
for computing the product of a vector with the tensor product 
of $n$ matrices of size $2\times 2$. The latter can in fact be
directly used to obtain a nucleation process for $(p,q)$-disjoint 
summation, even if an inefficient one.
(For an exposition of Yates's method we recommend Knuth 
\cite[\S4.6.4]{knuth1998art}; take $m_i=2$ and 
$g_i(s_i,t_i)=[\text{$s_i=0$ or $t_i=0$}]$ for $i=1,2,\ldots,n$
to extract the following nucleation process implicit in the
algorithm.)
For all $Z\subseteq [n]$ and $i\in\{0,1,\ldots,n\}$, let
\begin{equation}
\label{eq:yates-family}
a_i(Z)=
\{X\subseteq [n]:X\cap[n-i]=Z\cap[n-i],\ X\cap Z\setminus[n-i]=\emptyset\}\,.
\end{equation}
Put otherwise, $a_i(Z)$ consists of $X$ that agree with 
$Z$ in the first $n-i$ elements of $[n]$ and are disjoint from 
$Z$ in the last $i$ elements of $[n]$. In particular,
our objective is to assemble the sets 
$a_n(Y)=\{X:X\cap Y=\emptyset\}$ for each $Y\subseteq[n]$ starting 
from the singletons $a_0(X)=\{X\}$ for each $X\subseteq[n]$.
The nucleation process given by Yates' algorithm is, 
for all $i=1,2,\ldots,n$ and $Z\subseteq[n]$, to set
\begin{equation}
\label{eq:yates-nucleation}
a_i(Z)=\begin{cases}
a_{i-1}(Z\setminus\{n+1-i\})            & \text{if $n+1-i\in Z$},\\
a_{i-1}(Z\cup\{n+1-i\}) \cup a_{i-1}(Z) & \text{if $n+1-i\notin Z$}.
\end{cases}
\end{equation}
This results in $2^{n-1}n$ disjoint unions. If we restrict to 
the case $|Y|\leq q$ and $|X|\leq p$, then it suffices to consider 
only $Z$ with $|Z|\leq p+q$, which results in $O\bigl((p+q)\sum_{j=0}^{p+q}\binom{n}{j}\bigr)$ disjoint 
unions. Compared with our main result, this is not 
particularly efficient. In particular, our main result relies on 
``tree-projection'' partitioning that enables a significant 
speedup over the ``prefix-suffix'' partitioning in \eqref{eq:yates-family} 
and \eqref{eq:yates-nucleation}.

We observe that ``set nucleation'' can also be viewed as 
a computational problem, where the output collection is
given and the task is to decide whether there is a straight-line program 
of length at most $\ell$ that assembles the output using (disjoint) unions
starting from singleton sets. This problem is known to be NP-complete even 
in the case where output sets have size~$3$~\cite[Problem PO9]{garey-johnson};
moreover, the problem remains NP-complete if the unions are not required 
to be disjoint.

\section{A Circuit for $(p,q)$-Disjoint Summation}\label{section:circuit_construction}

\subsection{ Nucleation of $p$-subsets with a perfect binary tree.}
Looking at Valiant's circuit construction in the introduction, 
we observe that the left half of the circuit accumulates sums of variables 
(i.e., sums of 1-subsets of $[n]$) along what is a perfect binary tree. 
Our first objective is to develop a sufficient generalisation of this
strategy to cover the setting where each summand is indexed 
by a $p$-subset of $[n]$ with $p\geq 1$. 

Let us assume that $n=2^b$ for a nonnegative integer $b$ so that we 
can identify the elements of $[n]$ with binary strings of length $b$. 
We can view each binary string of length $b$ as traversing a unique 
path starting from the root node of a perfect binary tree 
of height $b$ and ending at a unique leaf node. 
Similarly, we may identify any node at level $\ell$ of the tree by 
a binary string of length $\ell$, with $0\leq \ell\leq b$. 
See~Fig.~\ref{figure:vector-tree}(a) for an illustration.
For $p=1$ this correspondence suffices.

\begin{figure}
\includegraphics{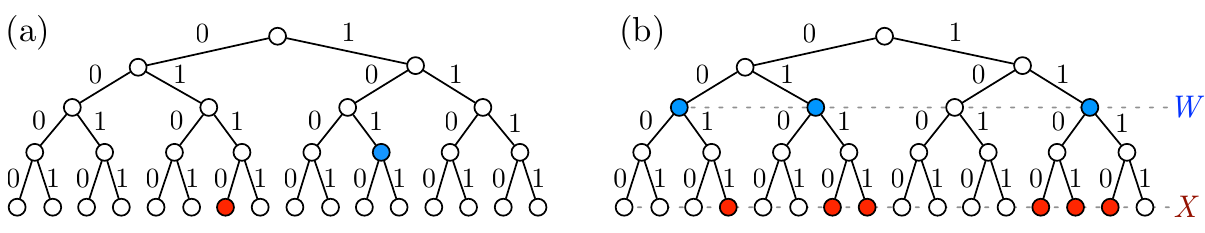}
\caption{Representing $\{0,1\}$-strings of length at most $b$ as nodes in a perfect binary tree of height $b$. Here $b=4$. (a)~Each string traces a unique path down from the root node, with the empty string $\epsilon$ corresponding to the root node. The nodes at level $0\leq\ell\leq b$ correspond to the strings of length $\ell$. The red leaf node corresponds to $0110$ and the blue node corresponds to $101$. (b)~A set of strings corresponds to a set of nodes in the tree. The set $X$ is displayed in red, the set $W$ in blue. The set $W$ is the projection of the set $X$ to level $\ell=2$. Equivalently, $X|_\ell=W$.}
\label{figure:vector-tree}
\end{figure}

For $p>1$, we are not studying individual binary strings of length $b$
(that is, individual elements of $[n]$), but rather $p$-subsets of such 
strings. In particular, we can identify each $p$-subset of $[n]$ 
with a $p$-subset of leaf nodes in the binary tree. To nucleate such
subsets it will be useful to be able to ``project'' sets upward in the
tree. This motivates the following definitions.

Let us write $\bv{\ell}$ for the set of all binary strings of 
length $0\leq\ell\leq b$. 
For $\ell = 0$, we write $\epsilon$ for the empty string.
For a subset $X \subseteq \{ 0, 1 \}^b$, 
we define the \emph{projection of $X$ to level $\ell$} as
\begin{equation}\label{eq:restriction}
X|_\ell = \Set{x \in \{ 0, 1\}^\ell}{ \exists y \in \{ 0, 1\}^{b-\ell} \text{ such that } xy \in X }\,.
\end{equation}
That is, $X|_\ell$ is the set of length-$\ell$ prefixes of strings in $X$.
Equivalently, in the binary tree we obtain $X|_\ell$ by lifting each 
element of $X$ to its ancestor on level $\ell$ in the tree. 
See~Fig.~\ref{figure:vector-tree}(b) for an illustration.
For the empty set we define $\emptyset|_\ell = \emptyset$.

Let us now study a set family $\F \subseteq 2^{\bv{b}}$. 
The intuition here is that each member of $\F$ is a summand, and $\F$ 
represents the sum of its members. A circuit design must assemble 
(nucleate) $\F$ by taking disjoint unions of carefully selected 
subfamilies. This motivates the following definitions.

For a level $0 \le \ell \le b$ and a string 
$W \subseteq \bv{\ell}$ let us define 
{\em the subfamily of $\F$ that projects to $W$} by 
\begin{equation}\label{eq:partition}
\F_W = \Set{X \in \F}{ X|_\ell = W}\,.	
\end{equation}
That is, the family $\F_W$ consists of precisely those members
$X\in\F$ that project to $W$. Again Fig.~\ref{figure:vector-tree}(b)
provides an illustration: we select precisely those $X$ whose 
projection is $W$.

The following technical observations are now immediate.
For each $0 \le \ell \le b$, if $\emptyset \in \F$, then we have
\begin{equation}\label{eq:part_emptyset}
\F_\emptyset = {\{ \emptyset \}}\,.
\end{equation}
Similarly, for $\ell = 0$ we have
\begin{equation}\label{eq:part_epsilon}
\F_{\{ \epsilon \}} = \F \setminus {\{ \emptyset \}}\,.
\end{equation}
For $\ell = b$ we have for every $W \in \F$ that
\begin{equation}\label{eq:tiling_base}
\F_W = \left\{ W \right\}\,.
\end{equation}

Now let us restrict our study to the situation where the family 
$\F \subseteq 2^\bv{b}$ contains only sets of size at most $p$. 
In particular, this is the case in our applications. 
For a set $U$ and an integer $p$, let us write 
$\setchoose{U}{p}$ for the family of all subsets of $U$ of size $p$,
and $\setchoose{U}{\downc p}$ for the family of all subsets 
of $U$ with size at most $p$. Accordingly, for integers $0\leq k\leq n$,
let us use the shorthand $\binom{n}{\downc k} = \sum_{i=0}^k\binom{n}{i}$.

The following lemma enables us to recursively nucleate any
family $\F\subseteq\bvc{b}{p}$. 
In particular, we can nucleate the family $\F_W$ 
with $W$ in level $\ell$ using the families $\F_Z$ with $Z$ in level 
$\ell+1$. Applied recursively, we obtain $\F$ by proceeding from the 
bottom up, that is, $\ell=b,b-1,\ldots,1,0$. 
The intuition underlying the lemma is illustrated in Fig.~\ref{fig:szw}.

\begin{figure}
\begin{center}
\includegraphics{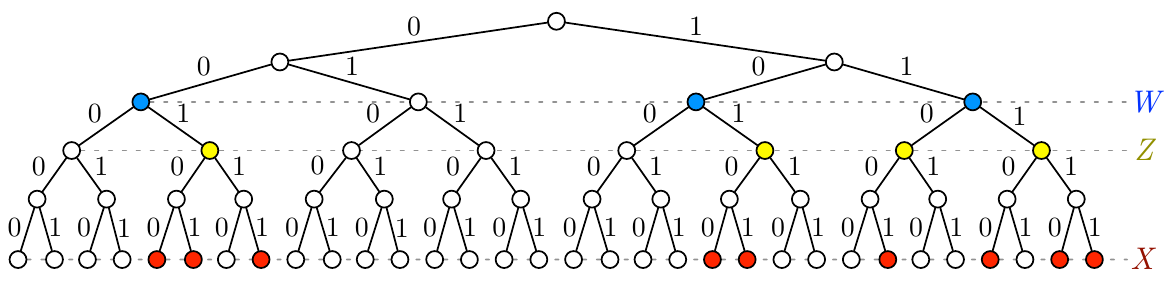}
\caption{Illustrating the proof of Lemma~\ref{lemma:tiling_composition}. 
Here $b=5$. The set $X$ (indicated with red nodes) projects to level $\ell=2$ to the set $W$ (indicated with blue nodes) and to level $\ell+1=3$ to the set $Z$ (indicated with yellow nodes). Furtermore, the projection of $Z$ to level $\ell$ is $W$. Thus, each $X\in\F$ is included to $\F_W$ exactly from $\F_Z$ in Lemma~\ref{lemma:tiling_composition}.}
\label{fig:szw}
\end{center}
\end{figure}

\begin{lemma}\label{lemma:tiling_composition}
For all\/ $0 \le \ell \leq b-1$, $\F\subseteq\bvc{b}{p}$, and $W \in \bvc{\ell}{p}$, we have that the family $\F_W$ is a disjoint union 
$\F_W=\bigcup\Set{\F_Z}{ Z \in \bvc{\ell+1}{p}_W}$.
\end{lemma}

\begin{proof}
The projection of each $X\in\F$ to level $\ell+1$ is unique, so the families $\F_Z$ are pairwise disjoint for distinct $Z$. Now consider an arbitrary $X \in \F$ and set $X|_{\ell+1} = Z$, that is, $X \in \F_Z$. From (\ref{eq:restriction}) we have $X|_\ell = Z|_\ell$, which implies that we have
$X \in \F_W$ if and only if 
$X|_\ell=W$ if and only if 
$Z|_\ell=W$ if and only if 
$Z \in \bvc{\ell+1}{p}_W$.
\end{proof}

\subsection{A generalisation: $(p,q)$-intersection summation.} 
It will be convenient to study a minor generalisation of $(p,q)$-disjoint summation. Namely, instead of insisting on disjointness, we allow nonempty intersections to occur with ``active'' (or ``avoided'') $q$-subsets $A$, but require that elements in the intersection of each $p$-subset and each $A$ are ``individualized.'' That is, our input is not given by associating a value $f(X)\in S$ to each set $X\in\setchoose{[n]}{\downc p}$, but is instead given by associating a value $g(I,X)\in S$ to each pair $(I,X)$ with $I\subseteq X\in\setchoose{[n]}{\downc p}$, where $I$ indicates the elements of $X$ that are ``individualized.'' In particular, we may insist (by appending to $S$ a formal identity element if such an element does not already exist in $S$) that $g(I,X)$ vanishes unless $I$ is empty. This reduces $(p,q)$-disjoint summation to the following problem:

\begin{problem}\label{problem:ssci}
$(${\bf $(p,q)$-intersection summation}$)$
Given as input a function $g$ that maps each pair $(I,X)$ with $I\subseteq X \in \setchoose{[n]}{\downc p}$ and $\card{I} \le q$ to an element $g(I,X) \in S$, output the function $h \colon \setchoose{[n]}{\downc q} \to S$ defined for all $A \in \setchoose{[n]}{\downc q}$ by
\begin{equation}\label{eq:ssci_target}
h(A) = \sgsum_{X \in \setchoose{[n]}{\downc p}} g(A \cap X,X)\,.
\end{equation}
\end{problem}

\subsection{The circuit construction.} 
We proceed to derive a recursion for the function $h$
using Lemma~\ref{lemma:tiling_composition} to carry out nucleation
of $p$-subsets.
The recursion proceeds from the bottom up, that is, $\ell=b,b-1,\ldots,1,0$
in the binary tree representation. (Recall that we identify the elements 
of $[n]$ with the elements of $\{0,1\}^b$, where $n$ is a power of 2 
with $n=2^b$.) The intermediate functions $h_\ell$ computed by 
the recursion are ``projections'' of \eqref{eq:ssci_target} 
using \eqref{eq:partition}. In more precise terms,
for $\ell = b,b-1,\ldots,1,0$, the function $h_\ell \colon \bvc{b}{q} \times \bvc{\ell}{p} \to S$ is defined for all $W \in \bvc{\ell}{p}$ and $A \in \bvc{b}{q}$ by
\begin{equation}\label{eq:h}
h_\ell(A, W) = \sgsum_{X \in \bvc{b}{p}_W} g(A \cap X, X)\,.
\end{equation}

Let us now observe that we can indeed recover the function $h$
from the case $\ell=0$.
Indeed, for the empty string $\epsilon$, the empty set $\emptyset$ 
and every $A \in \bvc{b}{q}$ we have 
by (\ref{eq:part_emptyset}) and (\ref{eq:part_epsilon}) that 
\begin{equation}\label{eq:output_ugly_hack}
h(A) = h_0(A, \{ \epsilon \}) \oplus h_0(A, \emptyset)\,.
\end{equation}

It remains to derive the recursion that gives us $h_0$. Here we require
one more technical observation, which enables us to narrow down the
intermediate values $h_\ell(A,W)$ that need to be computed to 
obtain $h_0$. In particular, we may discard the part of the active set $A$
that extends outside the ``span'' of $W$. This observation is the crux in
deriving a succinct circuit design.

For $0 \le \ell \le b$ and $w \in \bv{\ell}$, 
we define the \emph{span} of $w$ by
\[ 
\ind{w} = \Set{ x \in \bv{b} }{ \exists z \in \{ 0, 1\}^{b-\ell} \text{ such that } wz = x  }\,.
\]
In the binary tree, $\ind{w}$ consists of the leaf 
nodes in the subtree rooted at $w$.
Let us extend this notation to subsets $W \subseteq \bv{\ell}$ by
$\ind{W} = \bigcup_{w \in W} \ind{w}.$ 
The following lemma shows that it is sufficient to evaluate 
$h_\ell(A,W)$ only for $W \in \bvc{\ell}{p}$ and 
$A \in \bvc{b}{q}$ such that $A \subseteq \ind{W}$.

\begin{lemma}\label{lemma:limiting_h}
For all $0 \le \ell \le b$, $W \in \bvc{\ell}{p}$, and $A \in \bvc{b}{q}$, 
we have
\begin{equation}\label{eq:dp_limit}
h_\ell(A, W) = h_\ell(A \cap \ind{W}, W)\,.
\end{equation}
\end{lemma}

\begin{proof}
If $W \in \bvc{\ell}{p}$, $X \in \bvc{b}{p}$, and $X|_\ell = W$, 
then we have $X \subseteq \ind{W}$. Thus, directly from (\ref{eq:h}), we have
\[
\begin{split}
h_\ell(A, W) &= \sgsum_{X \in \bvc{b}{p}_W} g(A \cap X, X) \\
             &= \sgsum_{X \in \bvc{b}{p}_W} g((A \cap \ind{W}) \cap X, X) \\
             &= h_\ell(A \cap \ind{W}, W)\,. \qedhere
\end{split}
\]
\end{proof}

We are now ready to present the recursion for $\ell=b,b-1,\dots,1,0$.
The base case $\ell = b$ is obtained directly based on the values of $g$, because we have by (\ref{eq:tiling_base}) for all $W \in \bvc{b}{p}$ and $A \in \bvc{b}{q}$ with $A \subseteq W$ that
\begin{equation}\label{eq:dp_base}
h_b(A, W) = g(A, W)\,.
\end{equation}
The following lemma gives the recursive step from $\ell+1$ to $\ell$
by combining Lemma~\ref{lemma:tiling_composition} and
Lemma~\ref{lemma:limiting_h}.

\begin{lemma}\label{lemma:dp_ind}
For $0 \le \ell \le b-1$, $W \in \bvc{\ell}{p}$, and $A \in \bvc{b}{q}$ with $A \subseteq \ind{W}$, we have
\begin{equation}\label{eq:dp_ind}
h_\ell(A, W) = \sgsum_{Z \in \bvc{\ell+1}{p}_W} h_{\ell+1}(A \cap \ind{Z}, Z)\,.
\end{equation}
\end{lemma}

\begin{proof}
We have
\begin{align*}
    h_\ell(A, W) & = \sgsum_{X \in \bvc{b}{p}_W} g(A \cap X, X) &&  \text{ (\ref{eq:h})}\\
                 & = \sgsum_{Z \in \bvc{\ell+1}{p}_W} \ \sgsum_{X \in \bvc{b}{p}_Z} g(A \cap X, X) && \text{ (Lemma \ref{lemma:tiling_composition})}\\
                 & = \sgsum_{Z \in \bvc{\ell+1}{p}_W} h_{\ell+1}(A, Z) && \text{ (\ref{eq:h})}\\
                 & = \sgsum_{Z \in \bvc{\ell+1}{p}_W} h_{\ell+1}(A \cap \ind{Z}, Z)\,. && \text{ (Lemma \ref{lemma:limiting_h})}\qedhere
    \end{align*}
\end{proof}

\noindent
The recursion given by \eqref{eq:dp_base}, \eqref{eq:dp_ind}, and \eqref{eq:dp_limit} now defines an arithmetic circuit that solves $(p,q)$-intersection summation.

\subsection{\bf Size of the circuit.}
By (\ref{eq:dp_base}), the number of input gates in the circuit is equal to the number of pairs $(I, X)$ with $I \subseteq X \in \bvc{b}{p}$ 
and $\card{X} \le q$, which is
\begin{equation}\label{eq:no_ipgates}
\sum_{i=0}^p\sum_{j=0}^q \binom{2^b}{i}\binom{i}{j}\,.
\end{equation}

To derive an expression for the number of $\oplus$-gates, we count for each $0 \le \ell \le b-1$ the number of pairs $(A,W)$ with $W \in \bvc{\ell}{p}$, $A \in \bvc{b}{q}$, and $A \subseteq \ind{W}$, and for each such pair $(A,W)$ we count the number of $\oplus$-gates in the subcircuit that computes the value $h_\ell(A,W)$ from the values of $h_{\ell+1}$ using (\ref{eq:dp_ind}). 

First, we observe that for each $W \in \bvc{\ell}{p}$ we have $\card{\ind{W}} = 2^{b - \ell} \card{W}$. Thus, the number of pairs $(A, W)$ with $W \in \bvc{\ell}{p}$, $A \in \bvc{b}{q}$, and $A \subseteq \ind{W}$ is 
\begin{equation}\label{eq:no_awpairs}
\sum_{i=0}^p \sum_{j=0}^q \binom{2^\ell}{i} \binom{i 2^{b-\ell}}{j}\,.
\end{equation}
For each such pair $(A,W)$, the number of $\oplus$-gates for
(\ref{eq:dp_ind}) is $\card{\bvc{\ell +1}{p}_W}-1$.

\begin{lemma}
\label{lemma:no_intgates}
For all $0 \le \ell \le b-1$, $W \in \bvc{\ell}{p}$, and $\card{W} = i$, we have
\begin{equation}\label{eq:no_intgates}
	\card{\bvc{\ell+1}{p}_W} = \sum_{k = 0}^{p-i} \binom{i}{k} 2^{i-k}\,.
\end{equation}
\end{lemma}

\begin{proof}
A set $Z \in \bvc{\ell + 1}{p}_W$ can contain either one or both of the strings $w0$ and $w1$ for each $w \in W$. The set $Z$ may contain both elements for at most $p-i$ elements $w \in W$ because otherwise $\card{Z} > p$. Finally, for each $0 \le k \le p-i$, there are $\binom{i}{k}2^{i-k}$ ways to select a set $Z \in \bvc{\ell + 1}{p}_W$ such that $Z$ contains $w0$ and $w1$ for exactly $k$ elements $w \in W$.
\end{proof}

Finally, for each $A \in \bvc{b}{q}$ we require an 
$\oplus$-gate that is also designated as an output gate
to implement (\ref{eq:output_ugly_hack}).
The number of these gates is
\begin{equation}\label{eq:no_opgates}
\sum_{j=0}^q \binom{2^b}{j}\,.
\end{equation}

The total number of $\oplus$-gates in the circuit is obtained by combining (\ref{eq:no_ipgates}), (\ref{eq:no_awpairs}), (\ref{eq:no_intgates}), and (\ref{eq:no_opgates}). The number of $\oplus$-gates is thus
\begin{align*}
	& \sum_{i=0}^p\sum_{j=0}^q \binom{2^b}{i}\binom{i}{j} + \sum_{\ell = 0}^{b-1}\sum_{i=0}^p\sum_{j=0}^q \binom{2^\ell}{i}\binom{i 2^{b-\ell}}{j}\left(\sum_{k = 0}^{p-i} \binom{i}{k} 2^{i-k} - 1\right) + \sum_{j=0}^q \binom{2^b}{j}\\
  & \le \sum_{\ell = 0}^{b}\sum_{i = 0}^{p}\sum_{j=0}^{q}\binom{2^\ell}{i}\binom{i2^{b-\ell}}{j} 3^i \le \sum_{\ell = 0}^{b}\sum_{i = 0}^{p}\sum_{j=0}^{q}\frac{(2^\ell)^i}{i!}\frac{i^j(2^{b-\ell})^j}{j!}3^i\\
  & \le \sum_{\ell = 0}^{b}\sum_{i = 0}^{p}\sum_{j=0}^{q}\frac{(2^\ell)^{\max(p,q)}}{i!}\frac{i^j (2^{m-\ell})^{\max(p,q)}}{j!}3^i \\
  & = n^{\max(p,q)} (1 + \log_2 n) \sum_{i = 0}^{p}\sum_{j=0}^{q}\frac{i^j 3^i}{i!j!}\,.
\end{align*}
The double sum is at most a constant because we have that 
\begin{equation}\label{eq:constant1}
\sum_{i = 0}^{p}\sum_{j=0}^{q}\frac{i^j 3^i}{i!j!} \le \sum_{i = 0}^{\infty}\frac{3^i}{i!}\sum_{j = 0}^{\infty}\frac{i^j}{j!} = \sum_{i = 0}^{\infty}\frac{3^i}{i!}e^i \le \sum_{i = 0}^{\infty}\frac{(3e^2)^i}{i^i}\,,
\end{equation}
where the last inequality follows from Stirling's formula. Furthermore,
\begin{equation}\label{eq:constant2}
\sum_{i = \ceil{6e^2}}^{\infty}\frac{(3e^2)^i}{i^i} \le \sum_{i = \ceil{6e^2}}^{\infty}\frac{1}{2^i} \le \sum_{i = 0}^{\infty}\frac{1}{2^i} \le 2\,.
\end{equation}
Combining (\ref{eq:constant1}) and (\ref{eq:constant2}), we have
\[\sum_{i = 0}^{p}\sum_{j=0}^{q}\frac{i^j 3^i}{i!j!} \le \sum_{i = 0}^{\floor{6e^2}}\frac{(3e^2)^i}{i^i} + 2\,.\] 
Thus, the circuit defined in \S\ref{section:circuit_construction} has size $O((n^p + n^q) \log n )$, where the constant hidden by the $O$-notation does not depend on $p$ and $q$.

\subsection{\bf Constructing the Circuit.} In this section we give an algorithm that outputs the circuit presented above, given $b$, $p$, and $q$ as input. This algorithm can also be used to compute \eqref{eq:ssci_target} directly without constructing the circuit first.

\begin{algorithm}\label{alg:circuit_print}
Outputs a list of gates in the circuit, with labels on 
the input and output gates.
\begin{enumerate}[label=\arabic*.,ref=\arabic*]
        \item Initialise an associative data structure $D$
	\item For each $X \in \bvc{b}{p}$ and $I\in \setchoose{W}{\downc q}$, create an input gate $g$ labelled with $(I,X)$ and set $D(I,X) \gets g$.
	\item Set $\ell \gets b-1$.
	\item\label{step:l_loop_start} For each $W \in \bvc{\ell}{p}$ and $A \in \setchoose{\ind{W}}{\downc q}$, 
	\begin{enumerate}[label*=\arabic*.,ref=\arabic*]
		\item select an arbitrary $Z_0 \in \bvc{\ell + 1}{p}_W$ and set $g \gets D(A\cap\ind{Z_0},Z_0)$, 
	 	\item for each $Z \in \bvc{\ell + 1}{p}_W \setminus {\{ Z_0 \}}$ set $g \gets g \oplus D(A\cap\ind{Z},Z)$, and
		\item set $D(A, W) \gets g$.
	\end{enumerate}
	\item If $\ell \ge 1$, set $\ell \gets \ell -1$ and go to Step \ref{step:l_loop_start}.
	\item For each $A \in \bvc{b}{q}$, create an output gate $D(A, {\{\epsilon \}}) \oplus D(A, \emptyset)$ labelled with $A$.
\end{enumerate}
\end{algorithm}

Letting $k = \max(p,q)$, the sets that appear in Algorithm \ref{alg:circuit_print} can be represented in $O(k \log n)$ space, and each required operation on these sets can be done in $O(k \log n)$ time. Thus, we observe that iterating over set families in Algorithm \ref{alg:circuit_print} takes $O(k \log n)$ time per element, and the total number of iterations the algorithm makes is same as the number of gates in the circuit. Also, assuming that $D$ is a self-balancing binary tree, each search and insert operation takes $O{\bigl((k\log n)^2\bigr)}$ time, and a constant number of these operations is required for each gate. Thus, we have that the total running time of Algorithm \ref{alg:circuit_print} is $O{\bigl((p^2+q^2)(n^p + n^q) \log^3 n \bigr)}$.

\section{Concluding Remarks and Applications}
\label{sect:applications}

We have generalised Valiant's \cite{valiant1986negation} observation that negation is powerless for computing simultaneously the $n$ different disjunctions of all but one of the given $n$ variables: now we know that, in our terminology, subtraction is powerless for $(p, q)$-disjoint summation for any constant $p$ and $q$. (Valiant proved this for $p=q=1$.) Interestingly, requiring $p$ and $q$ be constants turns out to be essential, namely, when subtraction is available, an inclusion--exclusion technique is known \cite{bjorklundSUBMfourier} to yield a circuit of 
size $O\bigl(p\binom{n}{\downc p} + q \binom{n}{\downc q}\bigr)$, which, in terms of $p$ and $q$, is exponentially smaller than our bound $O\bigl((n^p + n^q)\log n\bigr)$. This gap highlights the difference of the algorithmic ideas behind the two results. Whether the gap can be improved to polynomial in $p$ and $q$ is an open question.

While we have dealed with the abstract notions of ``monotone sums'' or semigroup sums, in  applications they most often materialise as maximisation or minimisation, as described in the next paragraphs.  Also, in applications local terms are usually combined not only by one (monotone) operation but two different operations, such as ``$\min$'' and ``$+$''. To facilitate the treatment of such applications, we extend the semigroup to a semiring $(S, \oplus, \odot)$ by introducing a product operation ``$\odot$''. Now the task is to evaluate 
\begin{equation}
\label{eq:XY}
\bigoplus_{X,Y:X\cap Y=\emptyset} f(X)\odot g(Y)\,,
\end{equation}
where $X$ and $Y$ run through all $p$-subsets and $q$-subsets of $[n]$, respectively, and $f$ and $g$ are given mappings to $S$.
We immediately observe that the expression \eqref{eq:XY} is equal 
to $\bigoplus_{Y} e(Y)\odot g(Y)$, where the sum is over all $q$-subsets of $[n]$ and $e$ is as in \eqref{eq:e}. Thus, by our main result, it
can be evaluated using a circuit with $O((n^p + n^q)\log n)$ 
gates.

\subsection{Application to $k$-paths.}\label{section:k-paths}
We apply the semiring formulation to the problem of counting the maximum-weight $k$-edge paths from vertex $s$ to vertex $t$ in a given edge-weighted graph with real weights, where we assume that we are only allowed to add and compare real numbers and these operations take constant time (cf.~\cite{vassilevska2010finding}). By straightforward Bellman--Held--Karp type dynamic programming \cite{bellman1960combinatorial,bellman1962dynamic,held1962dynamic} (or, even by brute force) we can solve the problem in $\binom{n}{\downc k}n^{O(1)}$ time. However, our main result gives an algorithm that runs in $n^{k/2+O(1)}$ time by solving the problem in halves: Guess a middle vertex $v$ and define $f_1(X)$ as the number of maximum-weight $k/2$-edge paths from $s$ to $v$ in the graph induced by the vertex set $X\cup\{v\}$; similarly define $g_1(X)$ for the $k/2$-edge paths from $v$ to $t$. Furthermore, define $f_2(X)$ and $g_2(X)$ as the respective maximum weights and put $f(X) = (f_1(X), f_2(X))$ and $g(X) = (g_1(X), g_2(X))$. These values can be computed for all vertex subsets $X$ of size $k/2$ in $\binom{n}{k/2}n^{O(1)}$ time. It remains to define the semiring operations in such a way that the expression (\ref{eq:XY}) equals the desired number of $k$-edge paths; one can verify that the following definitions work correctly: $(c, w) \odot (c', w') = (c \cdot c', w+w')$ and 

\begin{equation*}
(c, w) \oplus (c', w') =
\begin{cases}
	(c, w) & \textrm{if $w>w'$},\\
	(c', w') & \textrm{if $w<w'$},\\
	(c+c', w) & \textrm{if $w=w'$}.	
\end{cases}
\end{equation*}

[[Please see Appendix~\ref{appendix:weight-count-semiring} for details.]]

Thus, the techniques of the present paper enable solving the problem essentially as fast as the fastest known algorithms for the special case of counting {\em all} the $k$-paths, for which quite different techniques relying on subtraction yield $\binom{n}{k/2}n^{O(1)}$ time bound \cite{bjorklund2009counting}. 
On the other, for the more general problem of counting weighted subgraphs Vassilevska and Williams \cite{vassilevska2009finding} give an algorithm whose running time, when applied to $k$-paths, is $O(n^{\omega k/3} + n^{2k/3 + c})$, where $\omega < 2.3727$ is the exponent of matrix multiplication and $c$ is a constant; this of course would remain worse than our bound even if $\omega = 2$.

\subsection{Application to matrix permanent.}
Consider the problem of computing 
the permanent of a $k \times n$ matrix $(a_{ij})$ over a {\em noncommutative semiring}, with $k \leq n$ and even for simplicity, given by 
$\sum_{\sigma} a_{1\sigma(1)} a_{2\sigma(2)}\cdots a_{k\sigma(k)}$,
where the sum is over all injective mappings $\sigma$ from $[k]$ to $[n]$. We observe that the expression (\ref{eq:XY}) equals the permanent if we let $p=q=k/2=\ell$ and define $f(X)$ as the sum of $a_{1\sigma(1)}a_{2\sigma(2)}\cdots a_{ \ell\sigma(\ell)}$ over all injective mappings $\sigma$ from $\{1,2,\ldots,\ell\}$ to $X$ and, similarly, $g(Y)$ as the sum of $a_{\ell+1\sigma(\ell+1)}a_{\ell+2\sigma(\ell+2)}\cdots a_{k\sigma(k)}$ over all injective mappings $\sigma$ from $\{\ell+1,\ell+2,\ldots,k\}$ to $Y$. Since the values $f(X)$ and $g(Y)$ for all relevant $X$ and $Y$ can be computed by dynamic programming in $\binom{n}{k/2}n^{O(1)}$ time, our main result yields the time bound $n^{k/2+O(1)}$ for computing the permanent.

Thus we improve significantly upon a Bellman--Held--Karp type dynamic programming algorithm that computes the permanent in $\binom{n}{\downc k}n^{O(1)}$ time, the best previous upper bound we are aware of for noncommutative semirings \cite{bjorklund2010evaluation}. It should be noted, however, that essentally as fast algorithms are already known for {\em noncommutative rings} \cite{bjorklund2010evaluation}, and that faster, $2^{k}n^{O(1)}$ time, algorithms are known for {\em commutative semirings} \cite{bjorklund2010evaluation,koutis2009limits}.

\subsection{Application to feature selection.}
The extensively studied feature selection problem in machine learning asks for a subset 
$X$ of a given set of available features $A$ so as to maximise some objective function $f(X)$. Often the size of $X$ can be bounded from above by some constant $k$, and sometimes the selection task needs to be solved repeatedly with the set of available features  $A$ changing dynamically across, say, the set $[n]$ of all features. Such constraints take place in a recent work \cite{decampos2011efficient} on Bayesian network structure learning by branch and bound: the algorithm proceeds by forcing some features, $I$, to be included in $X$ and some other, $E$, to be excluded from $X$. Thus the key computational step becomes that of maximising $f(X)$ subject to $I \subseteq X \subseteq [n]\setminus E$ and $|X|\leq k$, which is repeated for varying $I$ and $E$. We observe that instead of computing the maximum every time from scratch, it pays off precompute a solution to $(p, q)$-disjoint summation for all $0 \le p, q \le k$, since this takes about the same time as a single step for $I = \emptyset$ and any fixed $E$. Indeed, in the scenario where the branch and bound search proceeds to exclude each and every subset of $k$ features in turn, but no larger subsets, such precomputation decreases the running time bound quite dramatically, from $O(n^{2k})$ to $O(n^k)$; typically, $n$ ranges from tens to some hundreds and $k$ from $2$ to $7$. Admitted, in practice, one can expect the search procedure match the said scenario only partially, and so the savings will be more modest yet significant.

\subsection*{Acknowledgement} We thank Jukka Suomela for useful discussions.

\bibliographystyle{splncs03}
\bibliography{paper}

\clearpage

\appendix

\begin{center}
{\large\bf APPENDIX}
\end{center}

\section{Valiant's Construction and Generalisations}
\label{appendix:valiant}

This section reviews Valiant's \cite{valiant1986negation} circuit construction for $(p,q)$-disjoint summation in the case $p=q=1$. We also present minor generalisations to the case when either $p = 1$ or $q = 1$. For ease of exposition we follow the conventions and notation from Sect.~\ref{section:circuit_construction}. Accordingly, we assume that $n = 2^b$ for a nonnegative integer $b$ and identify the elements of $[n]$ with binary strings in $\bv{b}$.

\subsection{Valiant's construction} For $p = q = 1$, the disjoint summation problem reduces to the following form: given $f \colon \bv{b} \to S$ as input, compute $e:\bv{b}\rightarrow S$ defined for all $y,x\in\{0,1\}^b$ by 
\begin{equation}\label{eq:e11}
e(y) = \sgsum_{x:x \not= y} f(x)\,.
\end{equation}
Valiant's construction computes these sums by first computing intermediate functions $h^{+}_{\ell} \colon \bv{\ell} \to S$ defined for $\ell = 1, 2, \dotsc, b$ and $w \in \bv{\ell}$ by
\begin{equation}\label{eq:h+11}
h^{+}_\ell(w) =  \sgsum_{x \in \ind{w}} f(x)\,,
\end{equation}
and $h^{-}_\ell \colon \bv{\ell} \to S$ defined for $\ell = 1, 2, \dotsc, b$ and  $u \in \bv{\ell}$ by
\begin{equation}\label{eq:h-11}
h^{-}_\ell(u) = \sgsum_{x \in \bv{b} \setminus \ind{u}} f(x)\,.
\end{equation}
The solution to (\ref{eq:e11}) can then be recovered as $e(y) = h^-_b(y)$ for all $y \in \bv{b}$.

The values (\ref{eq:h+11}) can be computed for $\ell = b, b-1, \dotsc,2,1$ using the recurrence
\begin{equation}\label{eq:h+11_rec}
	\begin{aligned}
	h^{+}_b(x) & = f(x)\\
	h^{+}_\ell(w) & = h^{+}_{\ell+1}(w0) \oplus h^{+}_{\ell+1}(w1)\,.
	\end{aligned}
\end{equation}
Assuming that functions  $h^{+}_\ell$ have been computed for all $\ell$, the values (\ref{eq:h-11}) can then be computed for $\ell = 1, 2, \dotsc, b$ as
\begin{equation}\label{eq:h-11_rec}
	\begin{aligned}
	h^{-}_1(u) & = h^{+}_1(1-u)\\
	h^{-}_\ell(ui) & = h^{-}_{\ell-1}(u) \oplus h^{+}_{\ell}(u(1-i))\,.
	\end{aligned}
\end{equation}
The number of $\oplus$-gates to implement (\ref{eq:h+11_rec}) and (\ref{eq:h-11_rec}) as a circuit is exactly
\[ \sum_{i=1}^{b-1} 2^{i} + \sum_{i=2}^{b} 2^{i} = 3n - 6 = O(n)\,.\]

\subsection{Generalisation for $p = 1$ and $q > 1$} 
For $p = 1$ and $q > 1$, the $(p,q)$-disjoint summation problem reduces to 
computing 
\begin{equation}\label{eq:e1q}
e(Y) = \sgsum_{x:x \notin Y} f(x)\,,
\end{equation}
where $x\in\bv{b}$ and $Y\in \setchoose{\bv{b}}{q}$. 

To evaluate (\ref{eq:e1q}), we proceed analogously to the $p = q = 1$ case, first computing functions $h^{+}_{\ell} \colon \bv{\ell} \to S$ as defined in (\ref{eq:h+11}). The second set of intermediate functions now consists of functions $h^{-}_\ell \colon \bvc{\ell}{q} \to S$, defined for $\ell = 0, 1, \dotsc, b$ and $U \in \bvc{\ell}{q}$ by
\begin{equation}\label{eq:h-1q}
h^{-}_\ell(U) = \sgsum_{x \in \bv{b} \setminus \ind{U}} f(x)\,.
\end{equation}
Then we have $e(Y) = h^-_b(Y)$ for all $Y \in \setchoose{\bv{b}}{q}$.

The values $h^+_\ell$ are computed as in Valiant's construction.
For $\ell = 0, 1, \dotsc, b$ and $U \in \bvc{\ell}{q}$, define
\[ \hat{U} = \Set{ x(1-i) }{ xi \in U \text{ and } x(1-i) \notin U}\,. \]
Now we can evaluate (\ref{eq:h-1q}) using recurrence
\begin{equation}
\begin{aligned}
h^{-}_0(\{ \epsilon \}) & = 0\\
h^{-}_0(\emptyset) & = h^{+}_1(0) \oplus h^{+}_1(1)\\
h^{-}_\ell(U) & = h^-_{\ell-1}(U|_{\ell-1}) \oplus \sgsum_{x \in \hat{U}} h^+_\ell(x) \,.
\end{aligned}
\end{equation}

The number of $\oplus$-gates to evaluate the values $h^+_\ell$ is $n-2$. To determine the number of $\oplus$-gates to evaluate the values $h^-_\ell$, we note that $|\hat{U}| \le |U|$, and thus $|\hat{U}| \le q$ gates are used for any $U \in \bvc{\ell}{q}$. Thus, the total number of $\oplus$-gates is at most
\[1 + q\cdot \card{\bigcup_{\ell = 1}^b\bvc{\ell}{q}} = 1 + q\sum_{\ell = 1}^b \binom{2^\ell}{\downc q} = 1 + q\sum_{i = 0}^q \sum_{\ell = 1}^b \binom{2^\ell}{i}.\]
For positive integers $i$ and $\ell$ we have
\[ \binom{2^{b-\ell}}{i} = \frac{2^b}{2^\ell i} \binom{2^{b-\ell}-1}{i-1} \le \frac{2^b}{2^\ell i} \binom{2^b-1}{i-1} = \frac{1}{2^\ell} \binom{2^b}{i}\,.\]
Thus for positive integers $i$ it holds that
\[\sum_{\ell = 1}^b \binom{2^\ell}{i} = \sum_{\ell = 0}^{b-1} \binom{2^{b-\ell}}{i} \le \binom{2^b}{i} \sum_{\ell = 0}^{b-1} \frac{1}{2^\ell} \le 2\binom{2^b}{i}\,, \]
and hence
\[1 + q\sum_{i = 0}^q \sum_{\ell = 1}^b \binom{2^\ell}{i} \le 1 + q\left(b+ \sum_{i = 1}^q 2\binom{2^b}{i}\right) = O\left(q \binom{2^b}{\downc q} \right)\,. \]
That is, the total number of $\oplus$-gates used is $O\bigl(q \binom{n}{\downc q} \bigr)$.

\subsection{Generalisation for $p > 1$ and $q=1$} 

For $p > 1$ and $q = 1$, the $(p,q)$-disjoint summation problem reduces to computing 
\begin{equation}\label{eq:ep1}
e(y) = \sgsum_{X:y \notin X} f(X)\,,
\end{equation}
where $X\in \setchoose{\bv{b}}{p}$ and $y\in\bv{b}$. 

Now the first intermediate functions nucleate the inputs using tree-projection, that is, we define $h^+_\ell \colon \bvc{\ell}{p} \to S$ for $\ell = 0, 1, \dotsc, b$ and $W \in \bvc{\ell}{p}$ by
\begin{equation}\label{eq:h+p1}
h^+_\ell(W) = \sgsum_{X \in \bvc{b}{p}_W} f(X)\,.
\end{equation}
We define the second intermediate functions $h^-_\ell \colon \bv{\ell} \to S$ for all $\ell = 0, 1, \dotsc, b$ and $u \in \bv{\ell}$ by
\begin{equation}\label{eq:h-p1}
h^-_\ell(u) = \sgsum_{W:u \notin W} f(W)\,,
\end{equation}
where $W$ ranges over $\bvc{\ell}{p}$.
Again $e(y) = h^-_b(y)$ for all $y \in \bv{b}$.

By (\ref{eq:tiling_base}) and Lemma \ref{lemma:tiling_composition}, we can be compute (\ref{eq:h+p1}) by the recurrence
\begin{align*}
h^{+}_b( X ) & = f(X)\\
h^{+}_\ell(W) & = \sgsum_{Z \in \bvc{\ell+1}{p}_W} h^+_{\ell+1}(Z)\,.
\end{align*}
Similarly, we can compute (\ref{eq:h-p1}) by the recurrence
\begin{align*}
h^-_0(\{\epsilon\}) & = 0\\
h^-_\ell(xi) & = h^-_{\ell - 1}(x) \oplus \sgsum_{W \in \setchoose{\bv{\ell} \setminus \{ x0, x1 \}}{\downc (p-1)} } h^+_\ell (\{ x(1-i) \} \cup W)\,.
\end{align*}

The number of $\oplus$-gates to evaluate $h^+_\ell$ is $\binom{2^b}{\downc p} - 1$, and to evaluate $h^-_\ell$ at most
\begin{align*}
2 + \sum_{\ell=2}^{b}2^\ell \binom{2^\ell- 2}{\downc (p-1)} & = 2 + \sum_{\ell=2}^{b} \sum_{i=0}^{p-1} 2^\ell\binom{2^\ell- 2}{i}\\
& \le 2+ \sum_{\ell=2}^{b} \sum_{i=0}^{p-1} \frac{(i+1)2^\ell}{i+1}\binom{2^\ell- 1}{i}\\
& = 2+ \sum_{\ell=2}^{b} \sum_{i=0}^{p-1} (i+1)\binom{2^\ell}{i+1}\\
& \le 2+ p\sum_{\ell=2}^{b}\binom{2^\ell}{\downc p} = O\left(p\binom{2^b}{\downc p}\right)\,.
\end{align*}
Thus, the total number of $\oplus$-gates used
is $O\bigl(p\binom{n}{\downc p}\bigr)$.

\subsection{Generalisation for $p > 1$ and $q>1$} 

It is an open problem whether a circuit of size 
$O\bigl(p\binom{n}{\downc p}+q\binom{n}{\downc q}\bigr)$
exists when $p,q>1$.
Our main result in Sect.~\ref{section:circuit_construction} 
gives a circuit of size $O\bigl((n^p+n^q)\log n\bigr)$.

\section{Count-weight semirings}
\label{appendix:weight-count-semiring}

This section reviews the count-weight semiring used in counting heaviest $k$-paths as described in \S\ref{section:k-paths}. The following theorem is quite standard, but we give a detailed proof for the sake of completeness. 

\begin{theorem}
The Cartesian product $\N \times (\R \cup \{ -\infty \})$ is a commutative semiring when equipped with operations
\begin{equation}\label{eq:max_count_plus}
(c, w) \oplus (d, v) =
\begin{cases}
	(c, w) & \textrm{if $w>v$},\\
	(d, v) & \textrm{if $w<b$},\\
	(c+d, w) & \textrm{if $w=b$}	
\end{cases}
\end{equation}
and
\begin{equation}\label{eq:max_count_times}
(c, w) \odot (d, v) = (cd, w+v)\,.
\end{equation}
\end{theorem}

\begin{proof} In the following, we will use the \emph{Iverson bracket} notation; that is, if $P$ is a predicate we have $[P] = 1$ if $P$ is true and $[P] = 0$ if $P$ is false. For the maximum of two elements $x,y \in \R$, we write $x \lub y = \max(x,y)$. We note that $\lub$ is a associative, commutative binary operation on $\R \cup \{ - \infty \}$.
	
First, let us note that it follows directly from (\ref{eq:max_count_plus}) that
\[ (c, w) \oplus (d, v) = (c \, [w = w\lub v] + d \, [v = w \lub v],w \lub v)\,.\]
We now prove the claim using this observation and the known properties of $+$ and $\lub$.

\medskip
\noindent\emph{Associativity of  $\oplus$.}
The associativity of $\oplus$ follows from the associativity of $+$ and $\lub$, as we have
\begin{align*}
	&(c_1,w_1) \oplus \left((c_2,w_2) \oplus (c_3, w_3)\right)\\
	&\hspace{1cm} = (c_1,w_1) \oplus \bigl(c_2 \, [w_2 = w_2 \lub w_3] + c_3 \, [w_3 = w_2 \lub w_3], w_2 \lub w_3\bigr)\\
	&\hspace{1cm} = \Bigl(\sum_{i=1}^3 c_i \, [w_i = w_1 \lub (w_2 \lub w_3)], w_1 \lub (w_2 \lub w_3)\Bigr)\\
	&\hspace{1cm} = \Bigl(\sum_{i=1}^3 c_i \, [w_i = (w_1 \lub w_2) \lub w_3], (w_1 \lub w_2) \lub w_3\Bigr)\\
	&\hspace{1cm} = \bigl(c_1 \, [w_1 = w_1 \lub w_2] + c_2 \, [w_2 = w_1 \lub w_2], w_1 \lub w_2\bigr) \oplus (c_3, w_3)\\
	&\hspace{1cm} = \left((c_1,w_1) \oplus (c_2,w_2)\right) \oplus (c_3, w_3)\,.
\end{align*}

\noindent\emph{Commutativity of  $\oplus$.} The commutativity of $\oplus$ follows from the commutativity of $+$ and $\lub$, as we have
\begin{align*}
(c, w) \oplus (d, v) & = (c \, [w = w\lub v] + d \, [v = w\lub v], w\lub v)\\
					   & =  (d \, [v = w\lub v] + c \, [w = w\lub v], w\lub v)\\
					   & = (d, v) \oplus (c, w)\,.
\end{align*}

\noindent\emph{Existence of additive identity.}
We have that $(0, {- \infty})$ is the identity element for $\oplus$, as we have
\begin{align*}
(c,w) \oplus (0, {- \infty}) & = \bigl(c \, [w = w \lub {-\infty}] + 0 \, [{-\infty} = w \lub {-\infty}], w \lub {-\infty}\bigr)\\
							& = (c, w)\,.
\end{align*}

\noindent\emph{Associativity of  $\odot$.}
We have
\begin{align*}
	&(c,w) \oplus \bigl((d,v) \oplus (e, u)\bigr)\\
	&\hspace{1cm} = (c,w) \oplus \bigl(d e, v + u\bigr)\\
	&\hspace{1cm} = \bigl(c  (d  e), w + (v + u)\bigl)\\
	&\hspace{1cm} = \bigl((c  d)  e, (w + v) + u\bigr)\\
	&\hspace{1cm} = \bigl(c  d, w + v\bigr) \oplus (e, u)\\
	&\hspace{1cm} = \bigl((c,w) \oplus (d,v)\bigr) \oplus (e, u)\,.\\
\end{align*}

\noindent\emph{Commutativity of  $\odot$.}
We have 
\[ (c, w) \odot (d, v) = (c  d, w+v) = (d  c, v+w) = (d, v) \odot (c, w)\,. \]

\noindent\emph{Existence of multiplicative identity.}
The multiplicative identity element is $(1,0)$, since
\[ (c,w) \odot (1,0) = (c  1, w + 0) = (c,w)\,.\]

\noindent\emph{Distributivity.}
As $\odot$ is commutativity, it suffices to prove that multiplication from left distributes over addition. We have 
\begin{align*}
	&(c,w) \odot \bigl((d,v) \oplus (e,u)\bigr)\\
	&\hspace{1cm} = (c,w) \odot \bigl(d \, [v = v \lub u] + e \, [u = v \lub u], v \lub u\bigr)\\
	&\hspace{1cm} = \bigl(c \, (d \, [v = v \lub u] + e \, [u = v\lub u]), w + (v \lub u)\bigr)\\
	&\hspace{1cm} = \bigl(cd \, [v = v \lub u] + ce \, [u = v \lub u], (v + w)\lub(u + w)\bigr)\\
	&\hspace{1cm} = \bigl(cd \, [w + v = (w + v)\lub(w + u)] + ce \, [w + u = (w + v)\lub(w+ u)],\\
	&\hspace{1.6cm} (w + v)\lub(w + u)\bigr)\\
	&\hspace{1cm} = (c  d, w + v) \oplus (c  e, w + u)\\
	&\hspace{1cm} = \bigl((c, w) \odot (d,v)\bigr) \oplus \bigl((c, w) \odot (d,v)\bigr)\,.
\end{align*}

\noindent\emph{Annihilation in multiplication.}
Finally, we have that the additive identity element annihilates in multiplication, that is,
\[ (c,w) \odot (0, {-\infty}) = (c  0, w + {-\infty}) = (0, {-\infty})\,.\]
\end{proof}

\end{document}